\documentclass[12pt,sfbold,
]{paper}
\usepackage{amsfonts}
\usepackage{amsmath}
\usepackage{theorem}
\usepackage{setspace}
\usepackage{authordate1-4}
\usepackage{fancyhdr}
\usepackage{geometry}

\newtheorem{theorem}{Theorem}
\newtheorem{example}[theorem]{Example}

\newtheorem{definition}[theorem]{Definition}

\newtheorem{lemma}[theorem]{Lemma}

\newtheorem{problem}[theorem]{Problem}

\newtheorem{remark}[theorem]{Remark}

\begin{document}

\author{Edgar Delgado-Eckert\thanks{%
Centre for Mathematical Sciences, Technische Universit\"{a}t M\"{u}nchen,
Boltzmannstr.3, 85747 Garching, Germany. Email: edgar.delgado-eckert@mytum.de%
} \thanks{%
Pathology Department, Tufts University, 150 Harrison Av., Boston, MA 02111,
USA (\emph{correspondence address}).} \thanks{%
The author acknowledges support by a Public Health Service grant (RO1
AI062989) to David Thorley-Lawson at Tufts University, Boston, MA.}}
\title{Reverse engineering time discrete finite dynamical systems: A
feasible undertaking?}
\maketitle

\begin{abstract}
With the advent of high-throughput profiling methods, interest in reverse
engineering the structure and dynamics of biochemical networks is high.
Recently an algorithm for reverse engineering of biochemical networks was
developed by Laubenbacher and Stigler. It is a top-down approach using time
discrete dynamical systems. One of its key steps includes the choice of a
term order. The aim of this paper is to identify minimal requirements on
data sets to be used with this algorithm and to characterize optimal data
sets. We found minimal requirements on a data set based on how many terms
the functions to be reverse engineered display. Furthermore, we identified
optimal data sets, which we characterized using a geometric property called
"general position". Moreover, we developed a constructive method to generate
optimal data sets, provided a codimensional condition is fulfilled. In
addition, we present a generalization of their algorithm that does not
depend on the choice of a term order. For this method we derived a formula
for the probability of finding the correct model, provided the data set used
is optimal. We analyzed the asymptotic behavior of the probability formula
for a growing number of variables $n$ (i.e. interacting chemicals).
Unfortunately, this formula converges to zero as fast as $r^{q^{n}},$ where $%
q\in 
\mathbb{N}
$ and $0<r<1$. Therefore, even if an optimal data set is used and the
restrictions in using term orders are overcome, the reverse engineering
problem remains unfeasible, unless prodigious amounts of data are available.
Such large data sets are experimentally impossible to generate with today's
technologies.
\end{abstract}


\begin{keywords}
   Reverse engineering, data requirements, biochemical networks, time discrete dynamical systems, orthogonality
 \end{keywords}


\section{\protect\bigskip Introduction}

Since the development of multiple and simultaneous measurement techniques
such as microarray technologies, reverse engineering of biochemical and, in
particular, gene regulatory networks has become a more important problem in
systems biology . One well-known reverse engineering approach are the
top-down methods, which try to infer network properties based on the
observed global input-output-response. The observed input-output-response is
usually only partially described by available experimental data.

Depending on the type of mathematical model used to describe a biochemical
process, a variety of top-down reverse engineering algorithms have been
proposed \cite{DeJong}, \cite{Haeseleer}, \cite{Gardner}. Each modeling
paradigm presents different requirements relative to quality and amount of
the experimental data needed. Moreover, for each type of model, a suitable
mathematical framework has to be developed in order to study the performance
and limitations of reverse engineering methods. For any given modeling
paradigm and reverse engineering method it is important to answer the
following questions:

\begin{enumerate}
\item What are the minimal requirements on data sets?

\item Can data sets be characterized in such a way that "optimal" data sets
can be identified? (Optimality meaning that the algorithm performs better
using such a data set compared to its performance using other data sets.)
\end{enumerate}

The second question is related to the \emph{design of experiments} and
optimality is characterized in terms of \emph{quantity and quality} of the
data sets.

\cite{MR2086931} developed a top-down reverse engineering algorithm for the
modeling paradigm of time discrete finite dynamical systems. Herein, we will
refer to it\ as the LS-algorithm. They apply their method to biochemical
networks by modeling the network as a time discrete finite dynamical system,
obtained by discretizing the concentration levels of the interacting
chemicals to elements of a finite field. One of the key steps of the
LS-algorithm includes the choice of a term order. The modeling paradigm of
time discrete finite dynamical systems generalizes the Boolean approach \cite%
{Kauffman} (where the field only contains the elements $0$ and $1$).
Moreover, it is a special case of the paradigm described in \cite{Thomas}.

Some aspects of the performance of the LS-algorithm were studied by \cite%
{MR2265002} in a probabilistic framework.

In this paper we investigate the two questions stated above in the
particular case of the LS-algorithm. For this purpose, we developed a
mathematical framework\footnote{%
This framework is based on a general linear algebraic result stated in \cite%
{MR???????}.} that allows us to study the LS-algorithm in depth. Having
expressed the steps of the LS-algorithm in our framework, we were able to
provide concrete answers to both questions: First, we found minimal
requirements on a data set based on how many terms the functions to be
reverse engineered display. Second, we identified optimal data sets, which
we characterize using a geometric property called "general position".
Moreover, we developed a constructive method to generate optimal data sets,
provided a codimensional condition is fulfilled.

In addition, we present a generalization of the LS-algorithm that does not
depend on the choice of a term order. We call this generalization the \emph{%
term-order-free reverse engineering method. }For this method we derive a
formula for the probability of finding the correct model\footnote{%
We will give a precise definition of "correct model".}, provided the data
set used satisfies an optimality criterion. Furthermore, we analyze the
asymptotic behavior of the probability formula for a growing number of
variables $n$ (i.e. interacting chemicals). Unfortunately, this formula
converges to zero as fast as $r^{q^{n}},$ where $q\in 
\mathbb{N}
$ and $0<r<1$. Consequently, we conclude that even if an optimal data set is
used and the restrictions imposed by the use of term orders are overcome,
the reverse engineering problem remains unfeasible, unless experimentally
impracticable amounts of data are available. This result discouraged us from
including in this paper any computational and algorithmic aspects of the
term-order-free reverse engineering method.

In contrast to \cite{MR2265002}, we focus here on providing possible
criteria for the design of specific experiments instead of assuming that the
data sets are generated randomly. Moreover, we do not necessarily assume
that information about the actual number of interactions in the biochemical
network is available.

The organization of this article is the following:

Section 2 is devoted to the mathematical background: We briefly describe the
LS-algorithm and provide a mathematical framework to study it. Moreover, we
introduce the term-order-free reverse engineering method. We finish the
section with a clear formulation of the questions studied in this paper.
Section 3 presents rigorous results and some of their consequences. In
Section 4 we summarize our main results, discuss their consequences and
provide further conclusions.

To fully understand the technical details of our analysis, very basic
knowledge in linear algebra and algebra of multivariate polynomials is
required. We refer the interested reader to \cite{Golan} and \cite{MR1417938}%
.

\section{Mathematical background}

\subsection{A short description of the LS-algorithm}

In the modeling paradigm described by \cite{MR2086931}, a biological or
biochemical system described by $n$ varying quantities is studied by taking $%
m$ consecutive measurements of each of the interacting quantities. This
yields one time series%
\begin{equation*}
\mathbf{\vec{s}}_{1}=(s_{11},s_{12},...,s_{1n}),...,\mathbf{\vec{s}}%
_{m}=(s_{m1},s_{m2},...,s_{mn})
\end{equation*}%
Such series of consecutive measurements are repeated $t$ times starting from
different initial conditions, where the length $m_{k}$ of the series may
vary. At the end of this experimental procedure, several time series are
obtained:%
\begin{equation*}
\begin{array}{c}
\overrightarrow{\mathbf{s1}}_{1},...,\overrightarrow{\mathbf{s1}}_{m_{1}} \\ 
\vdots \\ 
\overrightarrow{\mathbf{sk}}_{1},...,\overrightarrow{\mathbf{sk}}_{m_{k}} \\ 
\vdots \\ 
\overrightarrow{\mathbf{st}}_{1},...,\overrightarrow{\mathbf{st}}_{m_{t}}%
\end{array}%
\end{equation*}%
Each point in a time series is a vector in $%
\mathbb{R}
^{n}.$ Time series are then discretized using a discretization algorithm
that can be expressed as a map%
\begin{equation}
D:%
\mathbb{R}
^{n}\rightarrow S^{n}  \label{DiscretMapping}
\end{equation}%
where the set $S$ is a finite field of cardinality $p:=\left\vert
S\right\vert $ (the cardinality of the field used is determined during the
discretization process). The discretized time series can be written as%
\begin{equation*}
\overrightarrow{\mathbf{dk}}_{1}:=D(\overrightarrow{\mathbf{sk}}_{1}),...,%
\overrightarrow{\mathbf{dk}}_{m_{k}}:=D(\overrightarrow{\mathbf{sk}}%
_{m_{k}}),\text{ }k=1,...,t
\end{equation*}%
One fundamental assumption made in their paper is that the evolution in time
of the discretized vectors obeys a simple rule, namely, that there is a
function%
\begin{equation*}
F:S^{n}\rightarrow S^{n}
\end{equation*}%
such that%
\begin{equation}
\overrightarrow{\mathbf{dk}}_{i+1}=F(\overrightarrow{\mathbf{dk}}_{i})\text{
for }i=1,...,m_{k}-1,\text{ }k=1,...,t  \label{Interp.Cond.}
\end{equation}%
\cite{MR2086931} call $F$ the transition function of the system. One key
ingredient in the LS-algorithm is the fact that the set $S$ is endowed with
the algebraic structure of a finite field. Under this assumption, the rule (%
\ref{Interp.Cond.}) reduces to a polynomial interpolation problem in each
component, i.e. for each $j\in \{1,...,n\}$%
\begin{equation}
\mathbf{dk}_{(i+1)j}=F_{j}(\overrightarrow{\mathbf{dk}}_{i})\text{ for }%
k=1,...,t,\text{ }i=1,...,m_{k}-1  \label{Comp.Interp.Cond}
\end{equation}%
The information provided by the equations (\ref{Comp.Interp.Cond})\ usually
underdetermines the function\linebreak $F_{j}:S^{n}\rightarrow S,$ unless
for all possible vectors $\vec{x}\in $ $S^{n},$ the values $F_{j}(\vec{x})$
are established by (\ref{Comp.Interp.Cond}). Indeed, any non-zero polynomial
function that vanishes on all the data inputs%
\begin{equation*}
X:=\{\overrightarrow{\mathbf{dk}}_{i}\text{ }|\text{ }k=1,...,t,\text{ }%
i=1,...,m_{k}-1\}
\end{equation*}%
could be added to a function satisfying the conditions (\ref%
{Comp.Interp.Cond}) and yield a different function that also satisfies (\ref%
{Comp.Interp.Cond}). Among all those possible solutions, the LS-algorithm
chooses the \textit{most parsimonious} interpolating polynomial function $%
F_{j}:S^{n}\rightarrow S$ according to some chosen term order. To generate
the most parsimonious function the algorithm first takes as input the
discretized time series and generates functions $f_{j},$ $j=1,...,n$ that
satisfy (\ref{Comp.Interp.Cond}) for each $j\in \{1,...,n\}$
correspondingly. Secondly, it takes a monomial order $<_{j}$ as input and
generates the normal form of $f_{j}$ with respect to the vanishing ideal $%
I(X)$ and the given order $<.$ For every $j\in \{1,...,n\},$ this normal
form is the output $F_{j}$ of the algorithm.\newline
We also refer to 2.1 in \cite{MR2265002} for another rigorous description of
the LS-algorithm.

\subsection{A mathematical framework to study the reverse engineering problem%
}

The mathematical framework presented here is based on a general result
stated in \cite{MR???????}. This framework will allow us to study the
LS-algorithm as well as a generalized algorithm of it that is independent on
the choice of term orders.\newline
We start with the original problem: Given a time-discrete dynamical system
over a finite field $S$ in $n$ variables%
\begin{equation*}
F:S^{n}\rightarrow S^{n}
\end{equation*}%
and a data set $X\subseteq S^{n}$ generated by iterating the function $F$
starting at one or more initial values, what are the chances of
reconstructing the function $F$ if the LS-algorithm or a similar algorithm
is applied using $X$ as input time series?\footnote{%
From an experimental point of view the following question arises: What is
the function $F$ in an experimental setting? Contrary to the situation when
models with an infinite number of possible states are reverse engineered
(see 1.2 in \cite{Ljung:1999:SIT}), there is a finite number of experiments
that could be, at least theoretically, performed to \textit{completely}
characterize the system studied. In this sense, even in an experimental
setting, there is an underlying function $F.$ The components of this
function is what \cite{MR2265002} called $h_{true}$.} Since the algorithms
studied here generate an output model $G:S^{n}\rightarrow S^{n}$ by
calculating every single coordinate function $G_{i}:S^{n}\rightarrow S$
separately, we will focus on the reconstruction of a single coordinate
function $F_{i}$ which we will simply call $f.$ We will use the notation $%
\mathbf{F}_{q}$ for a finite field of cardinality $q\in 
\mathbb{N}
.$ In what follows, we briefly review the main definitions and results
stated and proved in \cite{MR???????}:\newline
We denote the $q^{n}$-dimensional vector space of functions $g:\mathbf{F}%
_{q}^{n}\rightarrow \mathbf{F}_{q}$ with $F_{n}(\mathbf{F}_{q}).$ A basis
for $F_{n}(\mathbf{F}_{q})$ is given by all the monomial functions $%
\overrightarrow{x}^{\alpha }:=x_{1}^{\alpha _{1}}\cdot ...\cdot
x_{1}^{\alpha _{n}}$ where the exponents $\alpha _{i}$ are non-negative
integers satisfying $\alpha _{i}<q.$ The set of all those monomial functions
is denoted with $(g_{nq\alpha })_{\alpha \in M_{q}^{n}},$ where $%
M_{q}^{n}:=\left\{ \alpha \in \left( 
\mathbb{N}
_{0}\right) ^{n}\mid \alpha _{j}<q\text{ }\forall \text{ }j\in
\{1,...,n\}\right\} .$ We call those monomial functions \emph{fundamental
monomial functions}.

\begin{theorem}[and Definition]
Let $\mathbf{F}_{q}$ be a finite field and $n,m\in 
\mathbb{N}
$ natural numbers with $m\leq q^{n}$. Further let%
\begin{equation*}
\vec{X}:=(\vec{x}_{1},...,\vec{x}_{m})\in (\mathbf{F}_{q}^{n})^{m}
\end{equation*}%
be a tuple of $m$ \textbf{different} $n$-tuples with entries in the field $%
\mathbf{F}_{q}.$ Then the mapping%
\begin{eqnarray*}
\Phi _{\vec{X}} &:&F_{n}(\mathbf{F}_{q})\rightarrow \mathbf{F}_{q}^{m} \\
f &\mapsto &\Phi _{\vec{X}}(f):=(f(\vec{x}_{1}),...,f(\vec{x}_{m}))^{t}
\end{eqnarray*}%
is a surjective linear operator. $\Phi _{\vec{X}}$ is called the \emph{%
evaluation epimorphism} \emph{of the tuple }$\vec{X}.$
\end{theorem}

For a given set $X\subseteq $ $\mathbf{F}_{q}^{n}$ of data points, the
interpolation problem of finding a function $g\in F_{n}(\mathbf{F}_{q})$
with the property%
\begin{equation*}
g(\vec{x}_{i})=b_{i}\text{ }\forall \text{ }i\in \{1,...,m\},\text{ }%
x_{i}\in X
\end{equation*}%
can be expressed using the evaluation epimorphism as: Find a function $g\in
F_{n}(\mathbf{F}_{q})$ with the property%
\begin{equation}
\Phi _{\vec{X}}(g)=\vec{b}  \label{Epimor.Interp.Cond.}
\end{equation}%
Since a basis of $F_{n}(\mathbf{F}_{q})$ is given by the fundamental
monomial functions $(g_{nq\alpha })_{\alpha \in M_{q}^{n}},$ the matrix%
\begin{equation*}
A:=(\Phi _{\vec{X}}(g_{nq\alpha }))_{\alpha \in M_{q}^{n}}\in M(m\times
q^{n};\mathbf{F}_{q})
\end{equation*}%
representing the evaluation epimorphism $\Phi _{\vec{X}}$ of the tuple $\vec{%
X}$ with respect to the basis $(g_{nq\alpha })_{\alpha \in M_{q}^{n}}$ of $%
F_{n}(\mathbf{F}_{q})$ and the canonical basis of $\mathbf{F}_{q}^{m}$ has
always the full rank $m=\min (m,q^{n}).$ That also means, that the dimension
of the $\ker (\Phi _{\vec{X}})$ is%
\begin{equation}
\dim (\ker (\Phi _{\vec{X}}))=\dim (F_{n}(\mathbf{F}_{q}))-m=q^{n}-m
\label{Nullity}
\end{equation}%
In the case $m<q^{n}$ where $m$ is strictly smaller than $q^{n}=\left\vert 
\mathbf{F}_{q}^{n}\right\vert $ we have $\dim (\ker (\Phi _{\vec{X}}))>0$
and the solution of the interpolation problem is not unique. There are
exactly $q^{\dim (\ker (\Phi _{\vec{X}}))}$ different solutions which
constitute an affine subspace of $F_{n}(\mathbf{F}_{q})$. Only in the case $%
m=q^{n},$ that means, when for all elements of $\mathbf{F}_{q}^{n}$ the
corresponding interpolation values are given, the solution is unique. If the
problem is underdetermined and no additional information about properties of
the possible solutions is given, any algorithm attempting to solve the
problem has to provide a selection criterion to pick a solution among the
affine space of possible solutions. The LS-algorithm chooses the most
parsimonious interpolating polynomial function according to some chosen term
order. A more geometric approach to pick one solution would be to select the
solution that is perpendicular (or orthogonal) to the affine space of
solutions. As stated in\ Remark and Theorem 32 of \cite{MR???????}, the
solution selected by the LS-algorithm is precisely the orthogonal solution.
For orthogonality to apply, a generalized inner product has to be defined on
the space $F_{n}(\mathbf{F}_{q}).$ We finish this subsection reviewing this
concepts (cf. \cite{MR???????}).\newline
The space $F_{n}(\mathbf{F}_{q})$ is endowed with a symmetric bilinear form $%
\left\langle \cdot ,\cdot \right\rangle :V\times V\rightarrow K,$ i.e. a
generalized inner product. Orthogonality and orthonormality are defined as
in an Euclidean vector space.\newline
For a given set $X\subseteq $ $\mathbf{F}_{q}^{n}$ of data points, consider
the evaluation epimorphism $\Phi _{\vec{X}}$ of the tuple $\vec{X}$ and its
kernel $\ker (\Phi _{\vec{X}}).$ Now, let $(u_{1},...,u_{s})$ be a basis of $%
\ker (\Phi _{\vec{X}})\subseteq F_{n}(\mathbf{F}_{q}).$ By the basis
extension theorem, we can extend the basis $(u_{1},...,u_{s})$ to a basis%
\begin{equation*}
(u_{1},...,u_{s},u_{s+1},...,u_{d})
\end{equation*}%
of the whole space $F_{n}(\mathbf{F}_{q}).$ (There are many possible ways
this extension can be performed. See more details below). As in example 6 of 
\cite{MR???????}, we can construct a generalized inner product on $F_{n}(%
\mathbf{F}_{q})$ by setting%
\begin{equation*}
\left\langle u_{i},u_{j}\right\rangle :=\delta _{ij}\text{ }\forall \text{ }%
i,j\in \{1,...,d\}
\end{equation*}%
The orthogonal solution of (\ref{Epimor.Interp.Cond.}) is the solution $%
v^{\ast }\in F_{n}(\mathbf{F}_{q})$ that is orthogonal to $\ker (\Phi _{\vec{%
X}}),$ i.e. it holds $\Phi _{\vec{X}}(v^{\ast })=\vec{b}$ and for an
arbitrary basis $(w_{1},...,w_{s})$ of $\ker (T)$ the following
orthogonality conditions hold%
\begin{equation*}
\left\langle w_{i},v^{\ast }\right\rangle =0\text{ }\forall \text{ }i\in
\{1,...,s\}
\end{equation*}%
The way we extend the basis $(u_{1},...,u_{s})$ of $\ker (\Phi _{\vec{X}})$
to a basis%
\begin{equation*}
(u_{1},...,u_{s},u_{s+1},...,u_{d})
\end{equation*}%
of the whole space $F_{n}(\mathbf{F}_{q})$ determines crucially the
generalized inner product we get by setting%
\begin{equation}
\left\langle u_{i},u_{j}\right\rangle :=\delta _{ij}\text{ }\forall \text{ }%
i,j\in \{1,...,d\}  \label{DefGenInnerProd}
\end{equation}%
Consequently, the orthogonal solution of $\Phi _{\vec{X}}(g)=\vec{b}$ may
vary according to the chosen extension $u_{s+1},...,u_{d}\in F_{n}(\mathbf{F}%
_{q}).$ In \cite{MR???????} a systematic way to extend the basis $%
(u_{1},...,u_{s})$ to a basis for the whole space is introduced. With the
basis obtained, the process of defining a generalized inner product
according to (\ref{DefGenInnerProd}) is called the \emph{standard
orthonormalization. }This is because the basis $%
(u_{1},...,u_{s},u_{s+1},...,u_{d})$ is orthonormal with respect to the
generalized inner product defined by (\ref{DefGenInnerProd}).\newline
As shown in Section 5 of \cite{MR???????}, using the generalized inner
product obtained by applying the standard orthonormalization, the functions
generated by the LS-algorithm are orthogonal solutions of the polynomial
interpolation problem as formulated in (\ref{Epimor.Interp.Cond.}). Under
these assumptions the orthogonal solution is also unique (see theorem 9 in 
\cite{MR???????}).\newline
The standard orthonormalization process depends on the way the elements of
the basis $(g_{nq\alpha })_{\alpha \in M_{q}^{n}}$ of fundamental monomial
functions are ordered. If they are ordered according to a term order, the
calculation of the orthogonal solution of (\ref{Epimor.Interp.Cond.}) yields
the same result as the LS-algorithm. If more general linear orders are
allowed, a more general algorithm emerges that is not restricted to the use
of term orders. This algorithm can be seen as a generalization of the
LS-algorithm. We call it the \emph{term-order-free reverse engineering
method. }The precise definition of the standard orthonormalization procedure
is stated in Section 4 of \cite{MR???????}. In the appendix we summarize the
steps of the term-order-free reverse engineering method.

\subsection{The questions studied in this paper}

The mathematical framework developed in the previous subsection will allow
us to answer the following questions regarding the LS-algorithm and its
generalization, the term-order-free reverse engineering method:

\begin{problem}
Given a function $f\in F_{n}(\mathbf{F}_{q}),$ what are the minimal
requirements on a set\linebreak $X\subseteq $ $\mathbf{F}_{q}^{n},$ such
that the LS-algorithm reverse engineers $f$ based on the knowledge of the
values that it takes on every point in the set $X$?
\end{problem}

\begin{problem}
Are there sets $X\subseteq $ $\mathbf{F}_{q}^{n}$ that make the LS-algorithm
more likely to succeed in reverse engineering a function $f\in F_{n}(\mathbf{%
F}_{q})$ based \emph{only} on the knowledge of the values that it takes on
every point in the set $X$?\footnote{%
A solution to this problem would provide criteria for the design of
experiments.}
\end{problem}

\begin{problem}
Given a function $f\in F_{n}(\mathbf{F}_{q})$ and an optimal set $X\subseteq 
$ $\mathbf{F}_{q}^{n}$ (in the sense of the previous problem). If the term
order used by the LS-algorithm is chosen randomly, can the probability of
success be calculated? If the linear order used by the term-order-free
method is chosen randomly, can the probability of success be calculated?
\end{problem}

\begin{problem}
What is the asymptotic behavior of the probability for a growing number of
variables $n$?
\end{problem}

It is pertinent to emphasize that, contrary to the scenario studied in \cite%
{MR2265002}, we do not necessarily assume that information about the number
of variables actually affecting $f$ is available. We will give further
comments on this issue at the end of the conclusions.

\section{Results}

\subsection{Basic definitions and facts}

For what follows recall that $M_{q}^{n}=\left\{ \alpha \in \left( 
\mathbb{N}
_{0}\right) ^{n}\mid \alpha _{j}<q\text{ }\forall \text{ }j\in
\{1,...,n\}\right\} .$

\begin{lemma}[and Definition]
Let $K$ be a field, $n,q\in 
\mathbb{N}
$ natural numbers and $K[\tau _{1},...,\tau _{n}]$ the polynomial ring in $n$
indeterminates over $K.$ Then the set of all polynomials of the form%
\begin{equation*}
\sum_{\alpha \in M_{q}^{n}}a_{\alpha }\tau _{1}^{\alpha _{1}}...\tau
_{n}^{\alpha _{n}}\in K[\tau _{1},...,\tau _{n}]
\end{equation*}%
with coefficients $a_{\alpha }\in K$ is a vector space over $K.$ We denote
this set with\linebreak $P_{q}^{n}(K)\subset K[\tau _{1},...,\tau _{n}].$
\end{lemma}

\begin{theorem}
\label{Phi}Let $\mathbf{F}_{q}$ be a finite field and $n\in 
\mathbb{N}
$ a natural number. Then the vector spaces $P_{q}^{n}(\mathbf{F}_{q})$ and $%
F_{n}(\mathbf{F}_{q})$ are isomorphic via the mapping%
\begin{eqnarray*}
\varphi &:&P_{q}^{n}(\mathbf{F}_{q})\rightarrow F_{n}(\mathbf{F}_{q}) \\
g &=&\sum_{\alpha \in M_{q}^{n}}a_{\alpha }\tau _{1}^{\alpha _{1}}...\tau
_{n}^{\alpha _{n}}\mapsto \varphi (g)(\vec{x}):=\sum_{\alpha \in
M_{q}^{n}}a_{\alpha }\overrightarrow{x}^{\alpha }
\end{eqnarray*}
\end{theorem}

\begin{definition}
Let $K$ be a field, $n,m\in 
\mathbb{N}
$ natural numbers and $K[\tau _{1},...,\tau _{n}]$ the polynomial ring in $n$
indeterminates over $K.$ Furthermore, let $g_{1},...,g_{m}\in K[\tau
_{1},...,\tau _{n}]$ be polynomials. The set%
\begin{equation*}
\left\langle g_{1},...,g_{m}\right\rangle :=\{h_{1}g_{1}+...h_{m}g_{m}\text{ 
}|\text{ }h_{1},...,h_{m}\in K[\tau _{1},...,\tau _{n}]\}
\end{equation*}%
is called the \emph{ideal} generated by $g_{1},...,g_{m}.$
\end{definition}

For a tuple $\vec{x}=(x_{1},...,x_{n})$ we write $x:=\{x_{1},...,x_{n}\}$
for the set containing all the entries in the tuple $\vec{x}.$

\subsection{Conditions on the data set}

\begin{definition}
Let $f\in F_{n}(\mathbf{F}_{q})$ be a polynomial function. The subset of $%
\mathbf{F}_{q}^{n}$ containing all values on which the polynomial function $%
f $ vanishes is denoted by%
\begin{equation*}
V(\varphi ^{-1}(f))
\end{equation*}%
where $\varphi $ is the mapping defined in theorem (\ref{Phi}).
\end{definition}

The following result tells us that if we are using the LS-algorithm to
reverse engineer a nonzero function we necessarily have to use a data set $X$
containing points where the function does not vanish.

\begin{theorem}
\label{TheoremNonZeroRHS}Let $f\in F_{n}(\mathbf{F}_{q})\backslash \{0\}$ be
a nonzero polynomial function. Furthermore let%
\begin{equation*}
\vec{X}:=(\vec{x}_{1},...,\vec{x}_{m})\in (\mathbf{F}_{q}^{n})^{m}
\end{equation*}%
be a tuple of $m$ different $n$-tuples with entries in the field $\mathbf{F}%
_{q}$, $\vec{b}\in \mathbf{F}_{q}^{m}$ be the vector defined by%
\begin{equation*}
b_{i}:=f(\vec{x}_{i}),\text{ }i=1,...,m
\end{equation*}%
and $v^{\ast }$the orthogonal solution of $\Phi _{\vec{X}}(g)=\vec{b}.$Then
if $v^{\ast }=f$ it follows\footnote{%
If $A$ is a set, $A^{c}$ denotes its complement}%
\begin{equation*}
V(\varphi ^{-1}(f))^{c}\cap X\neq \emptyset
\end{equation*}
\end{theorem}

\begin{proof}
If $V(\varphi ^{-1}(f))^{c}\cap X=\emptyset $ then by definition of $%
V(\varphi ^{-1}(f)),$ the vector $\vec{b}$ would be equal to the zero vector 
$\vec{0}.$ From Corollary 11 in Subsection 2.2 of \cite{MR???????} we know
that the orthogonal solution $v^{\ast }$of $\Phi _{\vec{X}}(g)=\vec{0}$ is
the zero function, thus $v^{\ast }\neq f.$
\end{proof}

\begin{theorem}
Let $f\in F_{n}(\mathbf{F}_{q})\backslash \{0\}$ be a nonzero polynomial
function. Furthermore let%
\begin{equation*}
\vec{X}:=(\vec{x}_{1},...,\vec{x}_{m})\in (\mathbf{F}_{q}^{n})^{m}
\end{equation*}%
be a tuple of $m$ different $n$-tuples with entries in the field $\mathbf{F}%
_{q}$, $\vec{b}\in \mathbf{F}_{q}^{m}$ be the vector defined by%
\begin{equation*}
b_{i}:=f(\vec{x}_{i}),\text{ }i=1,...,m
\end{equation*}%
and $v^{\ast }$the orthogonal solution of $\Phi _{\vec{X}}(g)=\vec{b}.$ In
addition, assume $V(\varphi ^{-1}(f))^{c}\cap X\neq \emptyset .$ Then it
holds%
\begin{equation*}
v^{\ast }=f\Leftrightarrow f\in span(u_{s+1},...,u_{d})
\end{equation*}
\end{theorem}

\begin{proof}
The claim follows directly from the definition of orthogonal solution and
its uniqueness.
\end{proof}

\begin{remark}
\label{RemarkMinimalAmountData}From the necessary and sufficient condition%
\begin{equation}
f\in span(u_{s+1},...,u_{d})  \label{condition}
\end{equation}%
it becomes apparent, that if the function $f$ is a linear combination of
more than\linebreak $d-s=m$ fundamental monomial functions, $f$ can not be
found as an orthogonal solution $v^{\ast }$ of $\Phi _{\vec{X}}(g)=\vec{b}.$
In particular, if $f$ is a linear combination containing all $d$ fundamental
monomial functions in $(g_{nq\alpha })_{\alpha \in M_{q}^{n}},$ no \emph{%
proper} subset $X\subset \mathbf{F}_{q}^{n}$ of $\mathbf{F}_{q}^{n}$ will
allow us to find $f$ as orthogonal solution of $\Phi _{\vec{X}}(g)=\vec{b}$
(where $b_{i}:=f(\vec{x}_{i}),$ $\vec{x}_{i}\in X$).
\end{remark}

\begin{remark}
From the condition (\ref{condition}) follows that it is \emph{necessary}
that a monomial function appearing in $f$ is linearly independent of the
basis vectors $u_{1},...,u_{s}$ of $\ker (\Phi _{\vec{X}}).$ For this
reason, the set $X$ should be chosen in such a way that no fundamental
monomial function $(g_{nq\alpha })_{\alpha \in M_{q}^{n}}$ is linearly
dependent on the basis vectors $u_{1},...,u_{s}$ of $\ker (\Phi _{\vec{X}}).$
Otherwise, some of the terms appearing in $f$ might vanish on the set $X$
and wouldn't be detectable by any reverse engineering method, \cite%
{MR2086931}. This problem introduces a more general question about the
existence of vector subspaces in \textquotedblleft general
position\textquotedblright :
\end{remark}

\begin{definition}
\label{DefGenPos}Let $W$ be a finite dimensional vector space over a finite
field $\mathbf{F}_{q}$ with\linebreak $\dim (W)=d>0.$ Furthermore, let $%
\left( w_{1},...,w_{d}\right) $ be a fixed basis of $W$ and $s\in 
\mathbb{N}
$ a natural number with $s<d.$ A vector subspace $U\subset W$ with $\dim
(U)=s$ is said to be in \emph{general position} with respect to the basis $%
\left( w_{1},...,w_{d}\right) $ if for any basis $(v_{1},...,v_{s})$ of $U$
and any injective mapping%
\begin{equation*}
\pi :\{1,...,(d-s)\}\rightarrow \{1,...,d\}
\end{equation*}%
the vectors%
\begin{equation*}
v_{1},...,v_{s},w_{\pi (1)},...,w_{\pi (d-s)}
\end{equation*}%
are linearly independent.
\end{definition}

It can be shown, that if the cardinality $q$ of the finite field $\mathbf{F}%
_{q}$ is sufficiently\ large, proper subspaces in general position of any
positive dimension always exist. The proof is provided in the appendix.

Now assume that $\ker (\Phi _{\vec{X}})$ is in general position with respect
to the basis $(g_{nq\alpha })_{\alpha \in M_{q}^{n}}$ of $F_{n}(\mathbf{F}%
_{q}).$ Following the basis extension theorem and due to the general
position of $\ker (\Phi _{\vec{X}})$, we can extend the basis $%
(u_{1},...,u_{s})$ of $\ker (\Phi _{\vec{X}})$ to a basis%
\begin{equation*}
(u_{1},...,u_{s},u_{s+1},...,u_{d})
\end{equation*}%
of the whole space $F_{n}(\mathbf{F}_{q}),$ where $\{u_{s+1},...,u_{d}\}%
\subset $ $\{g_{nq\alpha }\}_{\alpha \in M_{q}^{n}}$ is \emph{any} subset
with $d-s$ elements of $\{g_{nq\alpha }\}_{\alpha \in M_{q}^{n}}.$ Now we
can construct a generalized inner product on $F_{n}(\mathbf{F}_{q})$ by
setting%
\begin{equation*}
\left\langle u_{i},u_{j}\right\rangle :=\delta _{ij}\text{ }\forall \text{ }%
i,j\in \{1,...,d\}
\end{equation*}

The advantage in this situation is that there is no bias imposed by the data
on the monomial functions that can be used to extend the basis $%
(u_{1},...,u_{s})$ to a basis of $F_{n}(\mathbf{F}_{q}),$ i.e. there are no
restrictions on the structure of $\ker (\Phi _{\vec{X}})^{\perp }.$ In
addition, having this degree of freedom, it is possible to calculate the
exact probability of success of the method based on the number of
fundamental monomial functions actually contained in $f$. We will give an
explicit probability formula in the next Subsection. For our further
analysis we need the following intermediate result, whose proof is left to
the reader:

\begin{lemma}[and Definition]
Let $\mathbf{F}_{q}$ be a finite field, $n,s\in 
\mathbb{N}
$ natural numbers with\linebreak $s\leq \dim (F_{n}(\mathbf{F}_{q})).$
Furthermore, let $U\subset F_{n}(\mathbf{F}_{q})$ be an $s$-dimensional
subspace. Then the set%
\begin{equation*}
V(U):=V(\left\langle \varphi ^{-1}(u_{1}),...,\varphi
^{-1}(u_{s})\right\rangle )\subseteq \mathbf{F}_{q}^{n}
\end{equation*}%
where $(u_{1},...,u_{s})$ is any basis of $U$ is independent on the choice
of basis and it's called the \emph{variety of the subspace }$U.$
\end{lemma}

Now the following question arises: How should the set $X$ be chosen in order
to have $\ker (\Phi _{\vec{X}})$ in general position with respect to the
basis $(g_{nq\alpha })_{\alpha \in M_{q}^{n}}$? For a given natural number $%
s<d:=\dim (F_{n}(\mathbf{F}_{q}))$ the idea is to start from a basis $%
(u_{1},...,u_{s})$ of a vector subspace $U\subset F_{n}(\mathbf{F}_{q})$ in
general position with respect to the basis $(g_{nq\alpha })_{\alpha \in
M_{q}^{n}}.$ The next step is to calculate the variety 
\begin{equation*}
Y:=V(\left\langle \varphi ^{-1}(u_{1}),...,\varphi ^{-1}(u_{s})\right\rangle
)\subseteq \mathbf{F}_{q}^{n}
\end{equation*}%
We assume $Y\neq \emptyset $ and order its elements arbitrarily to a tuple $%
\vec{Y}:=(\vec{y}_{1},...,\vec{y}_{m})\in (\mathbf{F}_{q}^{n})^{m},$ where $%
m:=\left\vert Y\right\vert .$ We know from Remark 23 in Subsection 3.2 of 
\cite{MR???????} that $\dim (\ker (\Phi _{\vec{Y}}))=\dim (F_{n}(\mathbf{F}%
_{q}))-\left\vert Y\right\vert =d-m.$ Now, in general, for the kernel $\ker
(\Phi _{\vec{Y}})$ of the corresponding evaluation epimorphism $\Phi _{\vec{Y%
}}$ it holds%
\begin{equation*}
U\subseteq \ker (\Phi _{\vec{Y}})
\end{equation*}%
and therefore $s\leq \dim (\ker (\Phi _{\vec{Y}}))=d-m,$ i.e. $m\leq d-s$.
Now, the ideal scenario would be the case $\ker (\Phi _{\vec{Y}})=U,$ i.e. $%
m=d-s.$ A less optimistic scenario is given when $U\subset \ker (\Phi _{\vec{%
Y}})$ is a proper subspace of $\ker (\Phi _{\vec{Y}}).$ In such a situation,
ideally we would wish for $\ker (\Phi _{\vec{Y}})$ to be itself in general
position with respect to the basis $(g_{nq\alpha })_{\alpha \in M_{q}^{n}}$.
This issues raise the following question:\newline
When does there exist a subspace $U\subset F_{n}(\mathbf{F}_{q})$ in general
position with respect to the basis $(g_{nq\alpha })_{\alpha \in M_{q}^{n}}$
with $\dim (U)<\dim (F_{n}(\mathbf{F}_{q}))$ that in addition satisfies%
\begin{equation}
\left\vert V(\left\langle \varphi ^{-1}(u_{1}),...,\varphi
^{-1}(u_{s})\right\rangle )\right\vert =\dim (F_{n}(\mathbf{F}_{q}))-\dim (U)
\label{CodimBed}
\end{equation}%
This is an interesting question that requires further research. It is
related to whether the subspace $U$ is an ideal of $F_{n}(\mathbf{F}_{q}),$
when $F_{n}(\mathbf{F}_{q})$ is seen as an algebra with the multiplication
of polynomial functions as the multiplicative operation. In the Appendix we
provide examples in which two subspaces, both in general position, show a
different behavior regarding the condition (\ref{CodimBed}). We formalize
this property:

\begin{definition}
\label{codimCondition}Let $U\subset F_{n}(\mathbf{F}_{q})$ be a subspace and 
$(u_{1},...,u_{s})$ an arbitrary basis of $U$. $U$ is said to satisfy the 
\emph{codimension condition} if it holds%
\begin{equation*}
codim(U)=\left\vert V(\left\langle \varphi ^{-1}(u_{1}),...,\varphi
^{-1}(u_{s})\right\rangle )\right\vert
\end{equation*}%
where $codim(U):=\dim (F_{n}(\mathbf{F}_{q}))-\dim (U).$
\end{definition}

A subspace $U\subset F_{n}(\mathbf{F}_{q})$ in general position with respect
to the basis $(g_{nq\alpha })_{\alpha \in M_{q}^{n}}$ that satisfies the
codimension condition allows for the construction of an optimal set for use
with the LS-algorithm. The set $Y:=V(\left\langle \varphi
^{-1}(u_{1}),...,\varphi ^{-1}(u_{s})\right\rangle )$ (where $%
u_{1},...,u_{s} $ is a basis of $U$) has namely the property $\ker (\Phi _{%
\vec{Y}})=U,$ i.e. $\ker (\Phi _{\vec{Y}})$ is in general position with
respect to the basis $(g_{nq\alpha })_{\alpha \in M_{q}^{n}}.$ In other
words, subspaces in general position that satisfy the codimension condition
provide a basic component for a constructive method for generating optimal
data sets. More generally we define:

\begin{definition}
\label{OptimalSet}A set $X\subseteq \mathbf{F}_{q}^{n}$ such that $\ker
(\Phi _{\vec{X}})$ is in general position with respect to the basis $%
(g_{nq\alpha })_{\alpha \in M_{q}^{n}}$ is referred to as \emph{optimal}.
\end{definition}

\begin{remark}[and Definition]
\label{DefPseudoOpt}Additional study is required to prove whether optimal
data sets exist in general. (See the Appendix for concrete examples.)
However, if no optimal sets can be determined, it is still advantageous to
work with a data set $X$ that was obtained as $V(\left\langle \varphi
^{-1}(u_{1}),...,\varphi ^{-1}(u_{s})\right\rangle ),$ where $%
(u_{1},...,u_{s})$ is a basis for a subspace $U$ in general position with
respect to the basis $(g_{nq\alpha })_{\alpha \in M_{q}^{n}}.$ In this case,
at least $U\subseteq \ker (\Phi _{\vec{Y}})$ still holds and it might be
that the dimensional difference between $U$ and $\ker (\Phi _{\vec{Y}})$ is
small. We call such data sets \emph{pseudo-optimal}.\newline
\end{remark}

\subsection{Probabilities of finding the original function as the orthogonal
solution}

\begin{theorem}
\label{ThmProbForm}Let $\mathbf{F}_{q}$ be a finite field, $n,m\in 
\mathbb{N}
$ natural numbers with $m<\dim (F_{n}(\mathbf{F}_{q}))=:d.$ Furthermore, let 
$f\in F_{n}(\mathbf{F}_{q})\backslash \{0\}$ be a nonzero function
consisting of a linear combination of exactly $t$ fundamental monomial
functions and 
\begin{equation*}
\vec{X}:=(\vec{x}_{1},...,\vec{x}_{m})\in (\mathbf{F}_{q}^{n})^{m}
\end{equation*}%
a tuple of $m$ different $n$-tuples with entries in the field $\mathbf{F}%
_{q} $ such that $X$ is optimal. Now let $\vec{b}\in \mathbf{F}_{q}^{m}$ be
the vector defined as%
\begin{equation*}
b_{i}:=f(\vec{y}_{i}),\text{ }i=1,...,m
\end{equation*}%
$s:=\dim (\ker (\Phi _{\vec{X}}))=d-m$ (cf. (\ref{Nullity})), $\left(
u_{1},...,u_{s}\right) $ a basis for $\ker (\Phi _{\vec{X}})$ and\linebreak $%
\{u_{s+1},...,u_{d}\}\subset \{g_{nq\alpha }\}_{\alpha \in M_{q}^{n}}$ \emph{%
an arbitrary} subset containing $d-s$ elements. Then the probability $P$
that the orthogonal solution $g^{\ast }$ of $\Phi _{\vec{X}}(g)=\vec{b}$
with respect to the generalized inner product%
\begin{equation*}
\left\langle u_{i},u_{j}\right\rangle :=\delta _{ij}\text{ }\forall \text{ }%
i,j\in \{1,...,d\}
\end{equation*}%
fulfills $f=g^{\ast }$ is given by%
\begin{equation}
P=\frac{%
\begin{pmatrix}
q^{n}-t \\ 
q^{n}-m%
\end{pmatrix}%
}{%
\begin{pmatrix}
q^{n} \\ 
m%
\end{pmatrix}%
}\text{ if }t\leq m  \label{Prob.Formula}
\end{equation}%
and%
\begin{equation*}
P=0\text{ if }t>m
\end{equation*}
\end{theorem}

\begin{proof}
Due to the definition of general position, there are exactly%
\begin{equation*}
(d-s)!%
\begin{pmatrix}
\dim (F_{n}(\mathbf{F}_{q})) \\ 
\dim (F_{n}(\mathbf{F}_{q}))-s%
\end{pmatrix}%
=(d-s)!%
\begin{pmatrix}
d \\ 
d-s%
\end{pmatrix}%
=(d-s)!%
\begin{pmatrix}
q^{n} \\ 
m%
\end{pmatrix}%
\end{equation*}%
different ways to extend a basis $(u_{1},...,u_{s})$ of $U$ to a basis of $%
F_{n}(\mathbf{F}_{q})$ using $m=d-s$ fundamental monomial functions. If $%
t\leq m,$ among such extensions, only%
\begin{equation*}
(d-s)!%
\begin{pmatrix}
d-t \\ 
d-s-t%
\end{pmatrix}%
=(d-s)!%
\begin{pmatrix}
q^{n}-t \\ 
s%
\end{pmatrix}%
=(d-s)!%
\begin{pmatrix}
q^{n}-t \\ 
q^{n}-m%
\end{pmatrix}%
\end{equation*}%
use the $t$ fundamental monomial functions appearing in $f.$ Now (\ref%
{Prob.Formula}) follows immediately. If, on the other hand, $t>m,$ the
number of fundamental monomial functions usable to extend a basis $%
(u_{1},...,u_{s})$ of $\ker (\Phi _{\vec{X}})$ to a basis of $F_{n}(\mathbf{F%
}_{q})$ is too small and $\ker (\Phi _{\vec{X}})^{\perp }$ is not big enough
to generate $f.$
\end{proof}

\begin{remark}
\label{DificultiesProbForm}If the elements in the basis $(g_{nq\alpha
})_{\alpha \in M_{q}^{n}}$ are ordered in a decreasing way according to a
term order (the biggest element is at the left end, the smallest at the
right end and position $t$ means counting $t$ elements from the right to the
left) an analogous probability formula would be%
\begin{equation}
P=\frac{\text{Number of arrangements that place the mon. functions in }f%
\text{ after position }s}{\text{Total number of arrangements}}
\label{Mon.Prob.Formula}
\end{equation}%
where an arrangement is an order of the elements of $(g_{nq\alpha })_{\alpha
\in M_{q}^{n}}$ that obeys a term order. (Two different term orders could
generate the same arrangement of the elements in the \textit{finite} set $%
\{g_{nq\alpha }\}_{\alpha \in M_{q}^{n}}$). So, for instance, if $f$
contains a term involving the monomial function $x_{1}^{q-1}\cdot ...\cdot
x_{n}^{q-1},$ then the above probability (\ref{Mon.Prob.Formula}) would be
equal to zero, since every arrangement of the elements in $\{g_{nq\alpha
}\}_{\alpha \in M_{q}^{n}}$ that obeys a term order would make that monomial
function biggest. (It is inherent to term orders to make some monomial
functions \emph{always} biggest). In more general terms, it is difficult to
make estimates about the numbers involved in (\ref{Mon.Prob.Formula}). This
shows some of the disadvantages of using term orders.
\end{remark}

\begin{remark}
\label{RemAsymptBeh}Since for relatively small $n$ and $q$ the number $%
d:=q^{n}$ is already very large, it is obvious that one should calculate the
asymptotic behavior of the probability formula (\ref{Prob.Formula}) for $%
d\rightarrow \infty .$ Indeed, we have with $t\leq m$ 
\begin{eqnarray*}
0 &\leq &\tfrac{%
\begin{pmatrix}
d-t \\ 
d-m%
\end{pmatrix}%
}{%
\begin{pmatrix}
d \\ 
m%
\end{pmatrix}%
}=\dfrac{\dfrac{(d-t)!}{(d-m)!(m-t))!}}{\dfrac{d!}{m!(d-m)!}} \\
&=&\dfrac{(d-t)!m!}{(m-t)!d!}\leq \dfrac{(d-t)!m!}{d!} \\
&=&\frac{m!}{d(d-1)...(d-t+1)}\longrightarrow 0\text{ for }d\longrightarrow
\infty
\end{eqnarray*}%
If we write the amount of data used in proportion to the size $d=q^{n}$ of
the space $\mathbf{F}_{q}^{n},$ and the number of terms displayed by $f$
relative to the size $q^{n}$ of the basis $(g_{nq\alpha })_{\alpha \in
M_{q}^{n}},$ it becomes apparent how quickly the probability formula
converges to $0$ for $d\longrightarrow \infty .$ So let $r:=m/d$ and $\gamma
:=d-t.$ Then we would have%
\begin{eqnarray*}
\frac{P}{r^{d}} &=&\tfrac{%
\begin{pmatrix}
d-t \\ 
d-m%
\end{pmatrix}%
}{r^{d}%
\begin{pmatrix}
d \\ 
m%
\end{pmatrix}%
}=\dfrac{(d-t)!m!}{r^{d}(m-t)!d!} \\
&=&\dfrac{m(m-1)...(m-t+1)}{r^{d}d(d-1)...(d-t+1)}=\dfrac{rd(rd-1)...(rd-t+1)%
}{r^{d}d(d-1)...(d-t+1)} \\
&=&\dfrac{rdrd(1-\frac{1}{rd})...rd(1-\frac{t-1}{rd})}{r^{d}dd(1-\frac{1}{d}%
)...d(1-\frac{t-1}{d})}=\dfrac{r^{t}d^{t}(1-\frac{1}{rd})...(1-\frac{t-1}{rd}%
)}{r^{d}d^{t}(1-\frac{1}{d})...(1-\frac{t-1}{d})} \\
&=&\dfrac{r^{t}(1-\frac{1}{rd})...(1-\frac{t-1}{rd})}{r^{d}(1-\frac{1}{d}%
)...(1-\frac{t-1}{d})} \\
&=&\dfrac{r^{-\gamma }(1-\frac{1}{rd})...(1-\frac{t-1}{rd})}{(1-\frac{1}{d}%
)...(1-\frac{t-1}{d})}\longrightarrow r^{-\gamma }\text{ for }%
d\longrightarrow \infty
\end{eqnarray*}%
In particular, it holds%
\begin{equation*}
\frac{%
\begin{pmatrix}
d-t \\ 
d-rd%
\end{pmatrix}%
}{%
\begin{pmatrix}
d \\ 
rd%
\end{pmatrix}%
}\approx r^{t}\text{ for big }d
\end{equation*}%
This expression shows in a straightforward way how big the proportional
amount of data should be in order to have an acceptable confidence in the
obtained result. It also shows that for $t$ close to $d$ the probability is
very low and the reverse engineering not feasible. Usually no information
about $t$ is available, so it is advisable to work with the maximal $t,$
namely $d-1$ or with an average value for $t.$
\end{remark}

\section{Conclusions}

The results we have obtained in the previous section provide guidelines on
how to design experiments to generate data to be used with the LS-algorithm
for the purpose of reverse engineering a biochemical network.\newline
The following are minimal requirements on a set $X\subseteq $ $\mathbf{F}%
_{q}^{n},$ such that the LS-algorithm reverse engineers $f$ based on the
knowledge of the values that it takes on every point in the set $X:$

\begin{enumerate}
\item If the LS-algorithm is used to reverse engineer a nonzero function $%
f\in F_{n}(\mathbf{F}_{q})\backslash \{0\}$, necessarily the data set $X$
used must contain points were the function does not vanish. In other words,
not all the interpolation conditions must be of the type $\vec{x}_{i}\mapsto
0$ (Theorem \ref{TheoremNonZeroRHS}).

\item If the LS-algorithm is used to reverse engineer a function $f\in F_{n}(%
\mathbf{F}_{q})\backslash \{0\}$ displaying $t$ different terms, it requires 
\textbf{at least} $t$ different data points to \textit{completely} reverse
engineer $f$ (Remark \ref{RemarkMinimalAmountData}).

\item If $f\in F_{n}(\mathbf{F}_{q})\backslash \{0\}$ is a polynomial
function containing all $p^{n}$ possible fundamental monomial functions, no 
\emph{proper} subset $X\subset \mathbf{F}_{q}^{n}$ of $\mathbf{F}_{q}^{n}$
will allow the LS-algorithm to find $f$ (Remark \ref{RemarkMinimalAmountData}%
).
\end{enumerate}

Our results also make possible the identification of optimal sets $%
X\subseteq $ $\mathbf{F}_{q}^{n}$ that make the LS-algorithm more likely to
succeed in reverse engineering a function $f\in F_{n}(\mathbf{F}_{q})$ based
only on the knowledge of the values that it takes on every point in the set $%
X.$ Optimal data sets $X\subset \mathbf{F}_{q}^{n}$ are characterized by the
property that $\ker (\Phi _{\vec{X}})$ is in general position with respect
to the basis $(g_{nq\alpha })_{\alpha \in M_{q}^{n}}$ (see Definitions \ref%
{OptimalSet} and \ref{DefGenPos}). Their advantage is given by the fact that
they do not impose constraints on the set of candidate terms that can be
used to construct a solution. Summarizing we can say:

\begin{enumerate}
\item Even though such sets can be constructed in particular examples (see
Appendix), further research is required to prove their existence in general
terms.

\item If no optimal sets can be determined, it is still advantageous to work
with pseudo-optimal data sets (see Remark and Definition \ref{DefPseudoOpt}).
\end{enumerate}

Since the identified optimal data sets are sets $X\subset \mathbf{F}_{q}^{n}$
of discretized vectors, in a real application, the optimal data set $X$ has
to be transformed back to a corresponding set $\widetilde{X}\subset 
\mathbb{R}
^{n}$ of real vectors. This transformation can be performed using an
"inverse" function of the discretization mapping (\ref{DiscretMapping}).
This "inverse" function has to be defined by the user, given the fact that
discretization mappings are highly non-injective and by definition map
entire subsets $Z\subset 
\mathbb{R}
^{n}$ into a single value $\vec{z}\in \mathbf{F}_{q}^{n}$.

Having characterized optimal data sets, the next step in our approach was to
provide an exact formula for the probability that the LS-algorithm will find
the correct model under the assumption that an optimal data set is used as
input. As stated in Remark \ref{DificultiesProbForm}, we weren't able to
find such a formula for the LS-algorithm. The biggest difficulty we face is
related to the use of term orders inherent to the LS-algorithm. We overcome
this problem by considering a generalization of the LS-algorithm which we
call the term-order-free reverse engineering method (see Appendix). This
method not only allows for the calculation of the success probability but it
also eliminates the issues and arbitrariness linked to the use of term
orders (see Remark \ref{DificultiesProbForm}). In conclusion, our results on
this issue are:

\begin{enumerate}
\item It is still an open problem how to derive a formula for the success
probability of the LS-algorithm when optimal data sets are used as an input
and the term order is chosen randomly. As stated in Remark \ref%
{DificultiesProbForm}, one of the main problems here is related to the use
of term orders inherent to the LS-algorithm.

\item Let $f\in F_{n}(\mathbf{F}_{q})\backslash \{0\}$ be a nonzero function
consisting of the linear combination of exactly $t$ fundamental monomial
functions. If the linear order used by the term-order-free method is chosen
randomly, the probability of successfully retrieving $f$ using an optimal
data set $X$ of cardinality $\left\vert X\right\vert =m$ is given by (see
Theorem \ref{ThmProbForm})%
\begin{equation}
P=\frac{%
\begin{pmatrix}
q^{n}-t \\ 
q^{n}-m%
\end{pmatrix}%
}{%
\begin{pmatrix}
q^{n} \\ 
m%
\end{pmatrix}%
}\text{ if }t\leq m  \label{ProbForm2}
\end{equation}%
and%
\begin{equation*}
P=0\text{ if }t>m
\end{equation*}

\item Let $d=q^{n}$ be the cardinality of the space $\mathbf{F}_{q}^{n}.$
Furthermore, let $X$ be an optimal data set with cardinality $\left\vert
X\right\vert =m$ and $r:=m/d$ (note that $0<r<1$). Then the asymptotic
behavior of the probability formula (\ref{ProbForm2}) for $d\rightarrow
\infty $ (i.e. for $n\rightarrow \infty $) satisfies (see Remark \ref%
{RemAsymptBeh})%
\begin{equation*}
\frac{%
\begin{pmatrix}
d-t \\ 
d-rd%
\end{pmatrix}%
}{%
\begin{pmatrix}
d \\ 
rd%
\end{pmatrix}%
}\approx r^{t}\text{ for big }d
\end{equation*}
\end{enumerate}

As a consequence of the latter, we conclude that even if an optimal data set
is used and the restrictions imposed by the use of term orders are overcome,
the reverse engineering problem remains unfeasible, unless experimentally
impracticable amounts of data are available.

Finally, we comment on one scenario identified in \cite{MR2265002}.
Specifically, in Conclusion 4(a), \cite{MR2265002} makes the assumption that
the wiring diagram of each of the underlying functions is known, i.e. the
variables that actually affect the function $f$ are known. Under this
assumption, let $k$ be the number of variables affecting $f.$ If one could
perform specific experiments such that for all possible values that the $k$
variables can take the response of the network is measured, the function $f$
would be uniquely determined. In this situation, reverse engineering $f$
wouldn't imply making any choices among possible solutions. This raises the
question of how many measurements are needed and how big this data set would
be in proportion to the size $q^{n}$ of the space $\mathbf{F}_{q}^{n}$ of
all possible states the network can theoretically display. The number of
measurements needed is $q^{k}$ and therefore the proportion is equal to%
\begin{equation*}
\frac{q^{k}}{q^{n}}=\frac{1}{q^{n-k}}
\end{equation*}%
If $k$ is small compared to $n$ (which is generally assumed by \cite%
{MR2265002}), then the proportion would be conveniently small. In other
words, in \textit{relative} terms, it is worth performing the $q^{k}$
specific experiments. However, performing $q^{k}$ measurements might still
be beyond experimental feasibility.

\section{Acknowledgements}

We would like to thank Dr. Michael Shapiro for helpful comments and Dr.
Winfried Just for an informative and stimulating e-mail exchange. We are
grateful to Dr. Karen Duca and Dr. David Thorley-Lawson for their support.
We also would like to express our gratitude to Jill Roughan and Dr. Karen
Duca for proofreading the manuscript.

\section{Appendix}

\subsection{Examples of vector spaces in general position and the
codimension condition}

\begin{example}
Let $n=2$, $q=2$ and consider the vector space $F_{2}(\mathbf{F}_{2})$ an
its basis\linebreak $(g_{22\alpha })_{\alpha \in
M_{2}^{2}}=(x_{1}x_{2},x_{1},x_{2},1)$ ordered according to the
lexicographic order with $x_{1}>x_{2}.$ Furthermore let $%
U:=span(x_{1}x_{2}+x_{1}+x_{2}+1).$ The basis vector $%
u_{1}:=x_{1}x_{2}+x_{1}+x_{2}+1$ has the coordinates $(1,1,1,1)^{t}$ with
respect to the basis $(g_{22\alpha })_{\alpha \in M_{2}^{2}}$. Therefore, $U$
is in general position with respect to $(g_{22\alpha })_{\alpha \in
M_{2}^{2}}.$ It is easy to verify%
\begin{eqnarray*}
\left\vert V(\left\langle \varphi ^{-1}(u_{1})\right\rangle )\right\vert
&=&\left\vert \{(x,y)\in \mathbf{F}_{2}^{2}\text{ }|\text{ }xy+x+y+1=0\text{
mod }2\}\right\vert \\
&=&\left\vert \{(0,1),(1,0),(1,1)\}\right\vert =3 \\
&=&2^{2}-1=codim(U)
\end{eqnarray*}%
As a consequence, the set $X:=\{(0,1),(1,0),(1,1)\}$ constitutes an optimal
data set to reverse engineer any function $f\in F_{2}(\mathbf{F}_{2})$
displaying no more than $3$ terms. If the term-order-free reverse
engineering method is used, the probability of successfully retrieving a
nonzero function displaying $1$ term would be%
\begin{equation*}
P=\frac{%
\begin{pmatrix}
2^{2}-1 \\ 
2^{2}-3%
\end{pmatrix}%
}{%
\begin{pmatrix}
2^{2} \\ 
3%
\end{pmatrix}%
}=\frac{%
\begin{pmatrix}
3 \\ 
1%
\end{pmatrix}%
}{%
\begin{pmatrix}
4 \\ 
3%
\end{pmatrix}%
}=\frac{3}{4}=0.75
\end{equation*}%
For a function displaying $2$ terms $P=0.5$ and $3$ terms $P=0.25.$
\end{example}

\begin{example}
Let $n=2$, $q=3$ and consider the vector space $F_{2}(\mathbf{F}_{3})$ an
its basis\linebreak $(g_{23\alpha })_{\alpha \in
M_{3}^{2}}=(x_{1}^{2}x_{2}^{2},x_{1}^{2}x_{2},x_{1}x_{2}^{2},x_{1}^{2},x_{1}x_{2},x_{2}^{2},x_{1},x_{2},1) 
$ ordered according to a total degree term order with $x_{1}>x_{2}.$
Furthermore let $U$ be the $8$-dimensional subspace of $F_{2}(\mathbf{F}%
_{3}) $ generated by%
\begin{equation*}
U:=span(x_{1}^{2}x_{2}^{2}+x_{1}^{2}x_{2},x_{1}^{2}x_{2}+x_{1}x_{2}^{2},x_{1}x_{2}^{2}+x_{1}^{2},x_{1}^{2}+x_{1}x_{2},x_{1}x_{2}+x_{2}^{2},x_{2}^{2}+x_{1},x_{1}+x_{2},x_{2}+1)
\end{equation*}%
The coordinate vectors of the generating vectors are%
\begin{equation*}
\begin{array}{c}
\widehat{u}_{1}:=(1,1,0,...,0)^{t} \\ 
\widehat{u}_{2}:=(0,1,1,0,...,0)^{t} \\ 
\vdots \\ 
\widehat{u}_{8}:=(0,...,0,1,1)^{t}%
\end{array}%
\end{equation*}%
By calculating the determinant of the matrices%
\begin{equation*}
A_{j}:=%
\begin{pmatrix}
\widehat{u}_{1}^{t} \\ 
\widehat{u}_{2}^{t} \\ 
\vdots \\ 
\widehat{u}_{8}^{t} \\ 
e_{j}^{t}%
\end{pmatrix}%
,\text{ }j=1,...,9
\end{equation*}%
(where $e_{j}$ is the $j$th canonical unit vector of $\mathbf{F}_{3}^{9}$),
one can easily show that $U$ is in general position with respect to $%
(g_{23\alpha })_{\alpha \in M_{3}^{2}}.$ To determine the set $%
V(\left\langle \varphi ^{-1}(u_{1}),...,\varphi ^{-1}(u_{8})\right\rangle ),$
we start solving the three last equations given by%
\begin{equation*}
\begin{array}{c}
x_{2}^{2}+x_{1}=0 \\ 
x_{1}+x_{2}=0 \\ 
x_{2}+1=0%
\end{array}%
\Leftrightarrow 
\begin{array}{c}
x_{2}^{2}=-1 \\ 
x_{1}=1 \\ 
x_{2}=-1%
\end{array}%
\end{equation*}%
This system of equations has no solution in the set $\mathbf{F}_{3}^{2}.$
Therefore%
\begin{equation*}
V(\left\langle \varphi ^{-1}(u_{1}),...,\varphi ^{-1}(u_{8})\right\rangle
)=\emptyset
\end{equation*}%
Consequently, $U$ does not satisfy the codimension condition and thus does
not yield an optimal data set.
\end{example}

\subsection{Existence of vector subspaces in general position}

The proof is easy but quite technical. The basic idea of the proof is to
treat the problem over the real numbers and then construct a solution over
finite fields based on the existence of a solution over the real numbers.
This last step takes advantage of the density of the rational numbers in the
set of real numbers.

We recall the definition of general position for vector spaces over a finite
field:

\begin{definition}
Let $W$ be a finite dimensional vector space over a finite field $\mathbf{F}%
_{q}$ with\linebreak $\dim (W)=d>0.$ Furthermore, let $\left(
w_{1},...,w_{d}\right) $ be a fixed basis of $W$ and $s\in 
\mathbb{N}
$ a natural number with $s<d.$ A vector subspace $U\subset W$ with $\dim
(U)=s$ is said to be in \emph{general position} with respect to the basis $%
\left( w_{1},...,w_{d}\right) $ if for any basis $(v_{1},...,v_{s})$ of $U$
and any injective mapping%
\begin{equation*}
\pi :\{1,...,(d-s)\}\rightarrow \{1,...,d\}
\end{equation*}%
the vectors%
\begin{equation}
v_{1},...,v_{s},w_{\pi (1)},...,w_{\pi (d-s)}  \label{LinIndepCond}
\end{equation}%
are linearly independent.
\end{definition}

It can be easily shown that if the linear independence condition (\ref%
{LinIndepCond}) holds for one basis of $U$, it holds for every other basis
of $U.$

Now we will construct an $s$-dimensional subspace $U\subset W$ in general
position with respect to a given basis of $W$, where $s$ is an arbitrary
natural number with $s<d.$ For this purpose we will find the coordinates
with respect to the basis $\left( w_{1},...,w_{d}\right) $ of a basis of $U.$
We denote the sought coordinates as follows%
\begin{equation*}
\vec{\xi}_{1}=%
\begin{pmatrix}
x_{1} \\ 
\vdots \\ 
x_{d}%
\end{pmatrix}%
,\vec{\xi}_{2}=%
\begin{pmatrix}
x_{d+1} \\ 
\vdots \\ 
x_{2d}%
\end{pmatrix}%
,\cdots ,\vec{\xi}_{s}=%
\begin{pmatrix}
x_{(s-1)d+1} \\ 
\vdots \\ 
x_{sd}%
\end{pmatrix}%
\end{equation*}%
The next step is to count all different injective mappings $\pi
:\{1,...,(d-s)\}\rightarrow \{1,...,d\}$ as $\pi _{1},...,\pi _{N}.$ For
each $\pi _{i}$ we consider the coordinate vectors $\vec{\xi}_{1},...,\vec{%
\xi}_{s},\vec{e}_{\pi _{i}(1)},...,\vec{e}_{\pi _{i}(d-s)}$ with respect to
the basis $\left( w_{1},...,w_{d}\right) ,$ where $\vec{e}_{j}$ is the $j$th
canonical unit vector of $\mathbf{F}_{q}^{d}.$ Now, for $i=1,...,N$ we
define the determinant functions%
\begin{eqnarray*}
D_{\pi _{i}} &:&%
\mathbb{R}
^{sd}\rightarrow 
\mathbb{R}
\\
\vec{x} &\mapsto &\left\vert \vec{\xi}_{1},...,\vec{\xi}_{s},\vec{e}_{\pi
_{i}(1)},...,\vec{e}_{\pi _{i}(d-s)}\right\vert
\end{eqnarray*}%
where $\vec{e}_{j}$ is seen as the $j$th canonical unit vector of $%
\mathbb{R}
^{d}.$ The linear independence condition (\ref{LinIndepCond}) is equivalent
to%
\begin{equation*}
D_{\pi _{i}}(\vec{x})\neq 0
\end{equation*}%
Due to the structure of $\left( \vec{\xi}_{1},...,\vec{\xi}_{s},\vec{e}_{\pi
_{i}(1)},...,\vec{e}_{\pi _{i}(d-s)}\right) $ and by the Leibniz determinant
formula we know that $D_{\pi _{i}}$ are nonzero polynomial functions in the
variables $x_{1},...,x_{sd}$ and therefore nonzero analytic functions in $%
\mathbb{R}
^{sd}$ with infinite radius of convergence, (in particular, continuous
functions). Consequently, no $D_{\pi _{i}}$ can be identical to zero on any
open subset of $%
\mathbb{R}
^{sd}.$ By the continuity of $D_{\pi _{1}}$ we know that there is a
non-empty open subset $O_{1}\subseteq 
\mathbb{R}
^{sd}$ such that $D_{\pi _{1}}|_{O_{1}}\neq 0.$ Using the same argument we
know that there is a non-empty set $O_{2}\subseteq O_{1}$ open in $%
\mathbb{R}
^{sd}$ such that $D_{\pi _{2}}|_{O_{2}}\neq 0.$ After applying this argument 
$N$ times we identify a non-empty open subset $O_{N}\subseteq 
\mathbb{R}
^{sd}$ such that $D_{\pi _{i}}|_{O_{N}}\neq 0$ $\forall $ $i\in \{1,...,N\}.$
Since the set $%
\mathbb{Q}
^{sd}$ is a dense subset of $%
\mathbb{R}
^{sd},$ there is a point $\vec{y}\in O_{N}$ with rational entries, i.e. $%
y_{l}\in 
\mathbb{Q}
$ $\forall $ $l\in \{1,...,sd\}.$ Let%
\begin{equation*}
\vec{y}=(\frac{a_{1}}{b_{1}},...,\frac{a_{sd}}{b_{sd}})^{t}
\end{equation*}%
and $c:=\prod_{k=1}^{sd}b_{k}.$ Since $\vec{y}\in O_{N}$, we know $D_{\pi
_{i}}(\vec{y})\neq 0$ $\forall $ $i\in \{1,...,N\}$. By the rules of
determinants we also know%
\begin{equation*}
D_{\pi _{i}}(c\vec{y})\neq 0\text{ }\forall \text{ }i\in \{1,...,N\}
\end{equation*}%
Moreover, $c\vec{y}$ has integer entries, i.e. $cy_{l}\in 
\mathbb{Z}
$ $\forall $ $l\in \{1,...,sd\}.$ For a sufficiently large prime number $p$,
the entries $cy_{l}$ can be seen as elements of the finite field $\mathbf{F}%
_{p}$ of integers modulo $p.$ Therefore, the values $cy_{l}\in \mathbf{F}%
_{p},$ $l=1,...,sd$ can be used as the coordinates with respect to the basis 
$\left( w_{1},...,w_{d}\right) $ of a basis for an $s$-dimensional subspace $%
U\subset W$ in general position with respect to the basis $\left(
w_{1},...,w_{d}\right) $ of $W,$ a vector space over the finite field $%
\mathbf{F}_{p}.$%
\endproof%

\subsection{The term-order-free reverse engineering algorithm}

The input of the term-order-free reverse engineering algorithm is a set $%
X\subseteq $ $\mathbf{F}_{q}^{n}$ of $m\leq q^{n}$ different data points, a
list of $m$ interpolation conditions%
\begin{equation*}
\vec{x}_{i}\mapsto b_{i},\text{ }x_{i}\in X
\end{equation*}%
and a linear order $>$ for the elements of the basis $(g_{nq\alpha
})_{\alpha \in M_{q}^{n}}$ of $F_{n}(\mathbf{F}_{q}),$ (i.e. the elements of
the basis are ordered decreasingly according to $>$ ). The steps of the
algorithm are as follows:

\begin{enumerate}
\item Calculate the entries of the matrix%
\begin{equation*}
A:=(\Phi _{\vec{X}}(g_{nq\alpha }))_{\alpha \in M_{q}^{n}}\in M(m\times
q^{n};\mathbf{F}_{q})
\end{equation*}%
representing the evaluation epimorphism $\Phi _{\vec{X}}$ of the tuple $\vec{%
X}$ with respect to the basis $(g_{nq\alpha })_{\alpha \in M_{q}^{n}}$ of $%
F_{n}(\mathbf{F}_{q})$ and the canonical basis of $\mathbf{F}_{q}^{m}.$

\item Calculate a basis $\vec{y}_{1},...,\vec{y}_{s}\in \mathbf{F}_{q}^{d}$
of $\ker (A).$

\item Extend the basis $\vec{y}_{1},...,\vec{y}_{s}$ of $\ker (A)$ to a
basis $(\vec{y}_{1},...,\vec{y}_{s},\vec{y}_{s+1},...,\vec{y}_{d})$ of $%
\mathbf{F}_{q}^{d}$ using the standard orthonormalization procedure. (See
Section 4 of \cite{MR???????}).

\item Define a generalized inner product $\left\langle .,.\right\rangle :%
\mathbf{F}_{q}^{d}\rightarrow \mathbf{F}_{q}$ by setting%
\begin{equation*}
\left\langle \vec{y}_{i},\vec{y}_{j}\right\rangle :=\delta _{ij}\text{ }%
\forall \text{ }i,j\in \{1,...,d\}
\end{equation*}%
and calculate the entries of the matrix $S$ defined by%
\begin{equation*}
S_{ij}:=\left\langle \vec{e}_{i},\vec{e}_{j}\right\rangle ,\text{ }i,j\in
\{1,...,q^{n}\}
\end{equation*}%
where $\vec{e}_{j}$ is the $j$th canonical unit vector of $\mathbf{F}%
_{q}^{d}.$

\item The coordinate vector with respect to the basis $(g_{nq\alpha
})_{\alpha \in M_{q}^{n}}$ of the output function is obtained by solving the
following system of inhomogeneous linear equations%
\begin{eqnarray*}
A\vec{z} &=&\vec{b} \\
\vec{y}_{i}^{t}S\vec{z} &=&0,\text{ }i=1,...,s
\end{eqnarray*}
\end{enumerate}

The steps described above represent an intelligible description of the
algorithm and are not optimized for an actual computational implementation.

\bibliographystyle{authordate1}
\bibliography{MathRef}

\end{document}